\documentclass[a4paper]{elsarticle}
\usepackage[top=0.75in,bottom=0.75in,left=1.0in,right=1.0in]{geometry}
\usepackage{graphicx,amsthm,amsmath,amssymb,mathtools}
\usepackage{subcaption}
\usepackage[bitstream-charter]{mathdesign}
\usepackage[T1]{fontenc}
\usepackage{textcomp}
\usepackage{hyperref}
\hypersetup{colorlinks=true,
citecolor=blue,linkcolor=blue
}
\newtheorem{theorem}{Theorem}

\newtheorem{remark}{Remark}

\newcommand\norm[1]{\left\lVert#1\right\rVert}
\newcommand{\half}{\frac{1}{2}}

\makeatletter
\def\ps@pprintTitle{%
 \let\@oddhead\@empty
 \let\@evenhead\@empty
 \def\@oddfoot{\centerline{\thepage}}%
 \let\@evenfoot\@oddfoot}
\makeatother

\begin{document}

\begin{frontmatter}
\title{Structure Preserving Reduced Attitude Control of Gyroscopes}

\author{Nidhish Raj\fnref{fn1}}
\author{Leonardo Colombo\fnref{fn2}}
\author{Ashutosh Simha\fnref{fn3}}

\fntext[fn1]{N. Raj (nidhish.iitk@gmail.com) is with the Department of Aerospace Engineering, Indian Institute of Technology Kanpur, UP, 208016, India.}
\fntext[fn2]{L. Colombo (leo.colombo@icmat.es) is a member of the Instituto de Ciencias Matem\'aticas (CSIC-UAMUC3M-UCM), Calle Nicol\'as Cabrera 15, Campus UAM, Cantoblanco, 28049, Madrid, Spain.}
\fntext[fn3]{A. Simha (ashutosh.iisc@gmail.com) is a member of TTU Software Science Department,
Tallinn University of Technology, Tallinn, Estonia.}

\begin{abstract}
In this work, we design a reduced attitude controller for reorienting the spin axis of a gyroscope in a geometric control framework. The proposed reduced attitude controller preserves the inherent gyroscopic stability associated with a spinning axis-symmetric rigid body. The equations of motion are derived in two frames: a non-spinning frame to show the gyroscopic stability, and a body-fixed spinning frame for deriving the controller. The proposed controller is designed such that it retains the gyroscopic stability structure in the closed loop and renders the desired equilibrium almost-globally asymptotically stable. Due to the time critical nature of the control input, the controller is extended to incorporate the effect of actuator dynamics for practical implementation. Thereafter, a comparison in performance is shown between the proposed controller and a conventional reduced attitude geometric controller with numerical simulation. Finally, the controller is validated  experimentally on a spinning tricopter.

\end{abstract} 
\end{frontmatter}

\section{Introduction}
Control system design for spinning axis-symmetric rigid bodies (henceforth called gyroscopes) plays   an important role in mechanical systems such as gyroscopes, spinning satellites and spacecrafts, and underactuated multi-rotors. These systems are typically modeled as rigid bodies with constant spin about the body fixed axis of symmetry, and the spin rate is an order of magnitude higher than the angular velocity about the other axes. From a control design perspective, the main characteristic that distinguishes the dynamics of these systems from conventional rigid body systems is that, the gyroscopic torque contributes significantly to the overall system dynamics and consequently poses non-trivial challenges in control design. 

There have been several works on attitude control for reorienting rigid body systems via control torques such as \cite{shuster, enos1994optimal, spindler1996optimal, han2004simultaneous}. More recently, feedback laws for stabilizing and tracking the attitude and reduced attitude of rigid body, have been intrinsically developed without employing local coordinates (thereby mitigating their associated singularities \cite{bhat}) by exploiting the differential geometric structure of the underlying configuration manifold (see for instance \cite{bulloreduced, tlee, sphere, maithripala2006almost, leeso3}). In such a context, a geometric proportional-derivative (PD) control law is intrinsically designed such that the proportional error is the differential of a certain Morse function, and the derivative error is the difference between the velocity, and the desired velocity after transforming it via a transport map that is compatible with the proportional error function. A notable advantage of these controllers is that they facilitate aggressive, global, and agile maneuvering, and also provide exponential stability with almost-global domain of convergence. A related, but alternate approach which exploits the Hamiltonian nature of the equations of motion of a rigid body for designing stabilizing control laws using internal and external actuation can be found in \cite{bloch1992stabilization}. 

Reduced attitude control of spinning axis-symmetric rigid body is critical in problems such as spin-stabilized artillery shells \cite{zhao2016spin}, and spin-axis stabilization of satellites and spacecrafts \cite{lange1967controlspinaxis, biggs2016geometricspinaxis, frost1991attitudespinaxis, grahnboardspinaxis, westfall2015designspinaxis, sphere}. In this situation, the spin axis typically consists of devices such as warheads, solar sensors, antennas, telescopes, etc., which may be required to be pointed in a fixed direction, or maneuvered. Some additional work in this direction are on spinning satellites with flexible appendages \cite{flex1,flex2} and drag-free spinning satellites \cite{dragfree1,dragfree2}. 
Another class of problems are under-actuated multi-rotors such as tri-copters and other multi-rotors which experience actuator failures. These systems are rigid bodies with propellers oriented along a common body-fixed axis. The complete attitude of the vehicle is controlled by differential thrusts and torques generated by the propellers. However, the spinning rigid body control problem could arise here when the number of operational rotors may not be sufficient to control the complete attitude. For instance, in controlling a tri-copter or a quadrotor with a failed rotor, the thrust axis direction is controlled, whereas the angular velocity about this axis is uncontrolled, and typically saturates at a high rate. Some works in this direction are \cite{andrea, andrearelaxed, landing1, landing2, ramptricopter, high_speed, lpv, peng, lanzonjgcd}. In all the above works, either a linear controller for the reduced attitude is developed based on pole-placement, or a nonlinear control is developed based on cancelling nonlinearities such as gyroscopic moments, friction, etc. While these control laws may stabilize the reduced attitude, they may not be robust to  moment of inertia parameters, actuator dynamics and disturbances. As such, while canceling the nonlinear gyroscopic moment, one loses the advantage of inherent spin-axis stability of gyroscopes, which is crucial for stabilizing or regulating the reduced attitude. Moreover, when the gyroscopic term is large at high-spin rates, the error in modeling as well as cancellation error due to delayed sensor feedback and actuator dynamics is significantly magnified.  Another important factor that is neglected in these works is that, the gyroscopic effect completely changes the structure of the rigid body dynamics. For instance in order to maneuver the spin-axis of a gyroscope, it is well known that one must apply a torque about the axis perpendicular to the desired axis of rotation, and not about it, as in the case of non-spinning rigid bodies (see \cite{gyroawrejcewicz2012dynamics} for instance). Due to this fact, the existing control laws (which are basically adaptation of regular rigid body controllers) even in the absence of actuator dynamics or modeling uncertainties, are highly inefficient in controlling gyroscopic systems.

In this work, we develop a geometric control law for reorienting the spin axis of a gyroscope which preserves the gyroscopic stability in the closed-loop dynamics. The control law thereby enables efficient maneuvering and also exploits the preserved inherent gyroscopic stability. The spin dynamics are derived in the body-fixed frame, as well as a non-spinning frame (akin to the gimbal frame of the gyroscope). We show via phase-portrait analysis that under high spin conditions, the spin axis rotates on an average about an axis perpendicular to the axis of applied torque. Based on this fact, a reduced attitude controller is developed such that the error dynamics preserves the gyroscopic stability structure of the original spinning rigid body dynamics. The proposed control law is a geometric proportional-derivative law, which is almost-globally, locally exponentially stable and does not depend on moment of inertia parameters i.e. does not cancel gyroscopic terms. This control law is then modified to account for first order actuator dynamics, and is subsequently appended with an observer in order to avoid computation of angular acceleration, which is not directly obtained from sensors. 
The performance of the proposed structure preserving reduced attitude controller has been demonstrated via simulations and compared with a standard geometric controller. 
The advantages of the proposed controller over conventional ones that have been demonstrated in this paper are:
\begin{itemize}
    \item The proposed control law manuevers the spin axis along a near-optimal path to the desired reduced attitude (i.e. close to the geodesic on the sphere) while conventional controllers take a longer and highly non-optimal path.
    \item The proposed controller is robust to modeling errors, sensor feedback delays and actuator dynamics, while conventional controllers fail in the presence of these factors. 
    \end{itemize}
Further, the controller has also been experimentally validated on an axis-symmetric tri-copter, in order to establish its effectiveness, robustness and applicability. 

The paper is organized as follows. In Section 2, we derive the gyroscope dynamics and give a phase portrait analysis of the gyroscopic effect. Section 3 introduces the conventional reduced attitude controller and the proposed structure preserving controller. In Section 4 we give a comparison in performance of the two reduced attitude controllers with numerical simulation. Finally, in Section 5 we shown the experimental results, followed by concluding remarks. Technical details involved in the proofs have been deferred to an Appendix.

\section{Spinning rigid body dynamics}
In this section we derive the equations of motion in a non-spinning frame so as to highlight the gyroscopic structural properties and design the controller. Define the following three frames of references: $F_0$ inertial frame, $F_2$ body fixed frame, and an intermediate frame $F_1$ which has the same $Z$ axis as that of $F_2$ (see Fig. \ref{fig:fbd}) and zero spin  about this axis. 

\begin{figure}[h!]
\centering
\includegraphics[width=0.45\linewidth]{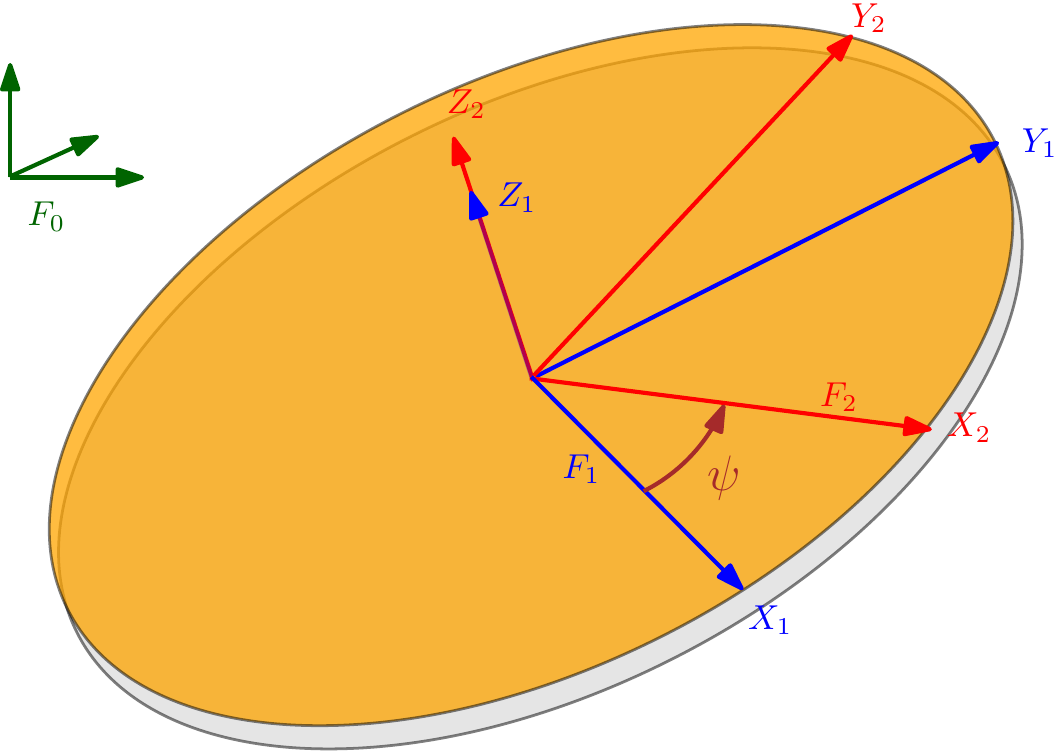}
\caption{Body-fixed spinning ($F_2$), body-fixed non-spinning ($F_1$), and inertial frame ($F_0$)}
\label{fig:fbd}
\end{figure}

The non-spinning frame $F_1$ and the body fixed frame $F_2$ are related to the inertial frame $F_0$ through the rotation matrices $R_1$ and $R_2$ respectively. $F_1$ and $F_2$ are offset by a rotation about $Z$ axis by an angle $\psi$ and those are related by $R_\psi = R_1^T R_2$. The motion of each individual frame is given by the following kinematic relation 
\begin{equation}
\dot{R}_i = R_i \hat{\omega}_i,
\end{equation} 
where $i \in \{ 1,2,\psi\}$. $\omega_1$ and $\omega_2$ are the angular velocity of frames $F_1$ and $F_2$ with respect to the inertial frame $F_0$. Here $\omega_\psi = [0,0,\bar{r}]^T$ with $\bar{r}$ the constant spin rate of the body and $\hat{\omega}_i$ is the skew-symmetric representation of $\omega_i=[\omega_{ix},\omega_{iy},\omega_{iz}]\in\mathbb{R}^{3}$, that is,
\begin{equation}
\hat{\omega}_i \triangleq \left[ {\begin{array}{ccc}
   0 & -\omega_{iz} &\omega_{iy}\\
   \omega_{iz} & 0 &-\omega_{ix}\\
   -\omega_{iy}& \omega_{ix}& 0\\
  \end{array} } \right].
\end{equation}

The angular velocities are related by 
\begin{equation}
\omega_1 = R_\psi \omega_2 - \omega_\psi.
\end{equation}
The moment of inertia of the body expressed in $F_2$ is given by $J_2 = diag(J_{sym}, J_{sym}, J_{zz})$ and is equal to its expression, $J_1$, in $F_1$ because of the axial symmetry of the body, i.e. $J_1 \triangleq R_\psi J_2 R_\psi^T = J_2$. Here $J_{sym}$ is the principle moment of inertia about the any axis in the symmetry plane and $J_{zz}$ is its value about the spin axis (Z-axis).

Since we are interested only in the orientation of the spin axis, the reduced attitude is defined by the orientation of the body frame Z-axis  expressed in the inertial frame, denoted $\Gamma \triangleq R_2 e_3 \in S^2 $, where $e_3 = [0,0,1]^T$. Similarly, the desired reduced attitude is given by $\Gamma_d \triangleq R_{2d} e_3$. Note that this definition is different form the one used in \cite{chaturvediAttTutorial}, wherein the reduced attitude is given by the body frame representation of an inertial frame fixed vector (i.e. $\Gamma = R_2^T e_3$). Therefore the reduced attitude kinematics is given by
\begin{equation}
\dot{\Gamma} = R_2 \hat{\omega}_2 e_3 = (R_2 \hat{\omega}_2 R_2^T) (R_2 e_3) = (R_2 \omega_2)^{\wedge} \times \Gamma = \Omega_2 \times \Gamma,
\end{equation}
where $\Omega_2\triangleq R_2\omega_2$.
Since $\Gamma$ represents the spin axis, the above equation could be simplified as 
\begin{equation} \label{eq:kin}
\dot{\Gamma} = \Omega_2 \times \Gamma = R_2 (\omega_2 \times e_3) = R_2 ((\bar{\omega} + \omega_{2z} e_3) \times e_3) = R_2 (\bar{\omega} \times e_3) = \Omega \times \Gamma,
\end{equation}
where $\Omega \triangleq R_2\bar{\omega}$, $\bar{\omega} \triangleq [\omega^T, 0]^T$, and $\omega \triangleq [\omega_{2x}, \omega_{2y}]^T \in \mathbb{R}^2 $.

The Euler's equation for the rigid body is given by 
\begin{equation}
J_2 \dot{\omega}_2 + \omega_2 \times J_2 \omega_2 = M_2,
\end{equation}
where $M_2 \triangleq [M_{2x}, M_{2y}, M_{2z}]$ is the net external torque acting on the body in frame $F_2$.
The above equation could be given in terms of the individual components as
\begin{equation}
\begin{aligned}
\dot{\omega}_{2x} &= -k \omega_{2y} + M_{2x}/J_{sym}, \\
\dot{\omega}_{2y} &= k \omega_{2x} + M_{2y}/J_{sym}, \\
\dot{\omega}_{2z} &= 0,
\end{aligned}
\end{equation}
where $k = (J_{zz} - J_{sym}) \bar{r}/J_{sym}$, with $\bar{r}$ the steady state spin rate about the Z-axis. At steady state, the net external torque about Z-axis, $M_{2z}$, is zero as the propeller drag torque is balanced by the aerodynamic drag torque due to rotation of the tricopter. The above equations for $\omega_x$ and $\omega_y$ could be compactly given by 
\begin{equation} \label{eq:gyro}
\dot{\omega} = A_{skw} \omega + u,
\end{equation}
where $u \triangleq [M_{2x}, M_{2y}]/J_{sym}$, and
\begin{equation}
 A_{skw} = \begin{bmatrix}
 0 & -k \\
 k & 0
 \end{bmatrix}.
\end{equation}

Next, we derive the above equations in the non-spinning frame, $F_1$, so as to highlight the gyroscopic effect. In order to do so, the expression for the angular momentum in $F_1$ is given by
\begin{equation}
H_1 = R_\psi J_2 \omega_2 = J_1 (\omega_1 + \omega_\psi).
\end{equation}
Taking the derivative of the above, the Euler's equation for the spinning body in frame $F_1$ is given by 
\begin{equation}
\dot{H}_1 = J_1 \dot{\omega}_1 + \omega_1 \times J_1 (\omega_1 + \omega_\psi) = M_1, 
\end{equation}
where $M_1$ is the external torque acting on the body expressed in frame $F_1$. 
The above equation when simplified results in a similar expression as in \eqref{eq:gyro}:
\begin{equation}
\dot{\bar{\omega}} = \bar{A}_{sk} \bar{\omega} + \bar{u},
\end{equation}
where where $\bar{\omega} = [\omega_{1x}, \omega_{1y}]$, $\bar{u} = [M_{1x}, M_{1y}]/J_{sym}$,
\begin{equation}
 \bar{A}_{sk} = \begin{bmatrix}
 0 & -\bar{k} \\
 \bar{k} & 0
 \end{bmatrix},
\end{equation}
and $\bar{k} = J_{zz}/J_{sym}\bar{r}$.
It is easy to note that the above equation is an undamped second-order linear system with natural frequency $\bar{k}$. A small amount of damping of the form $-k_D \bar{\omega}$ when added to the system would result in the following phase portraits with and without constant torques.

\begin{figure}[h!]
\centering
\begin{subfigure}[h!]{.3\textwidth}
\centering
	\includegraphics[width=1.1\textwidth]{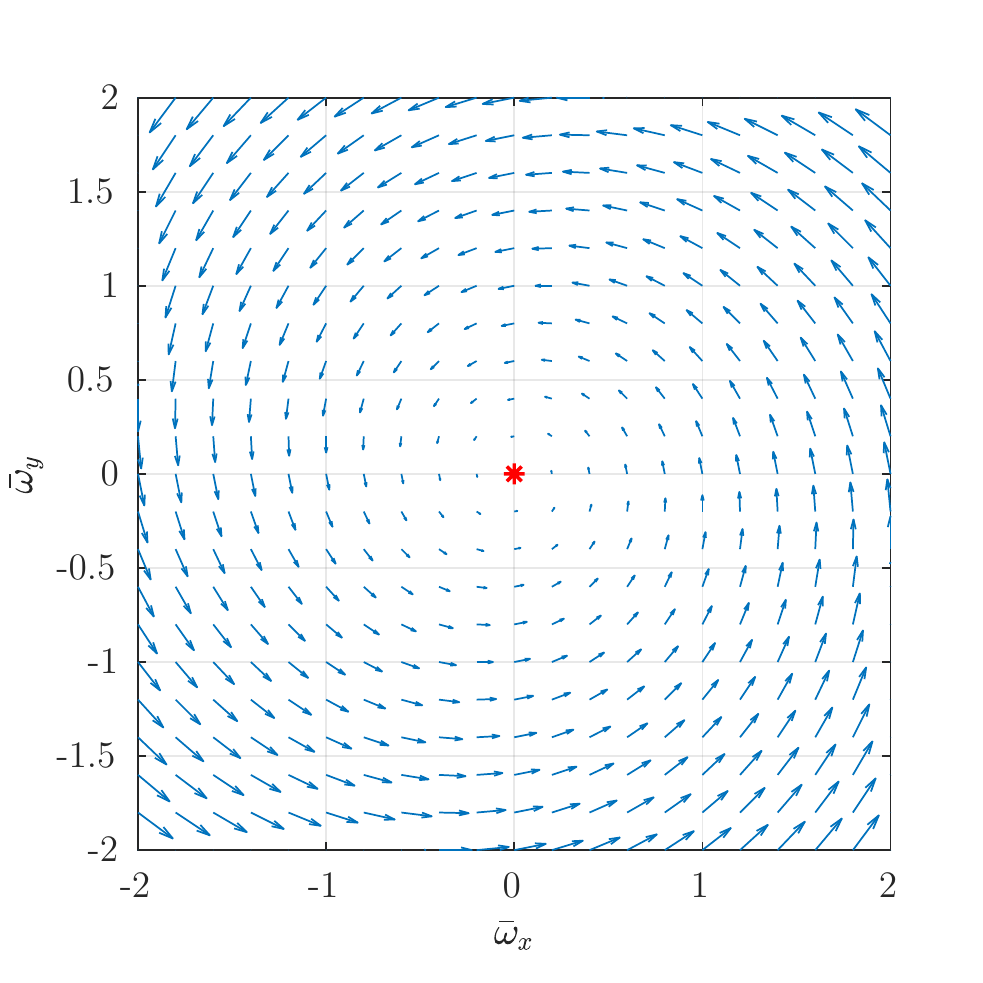}
	\caption{$\bar{u} = 0$}
	\label{fig:tor_0}
\end{subfigure} \qquad
\begin{subfigure}[h!]{.3\textwidth}
\centering
	\includegraphics[width=1.1\textwidth]{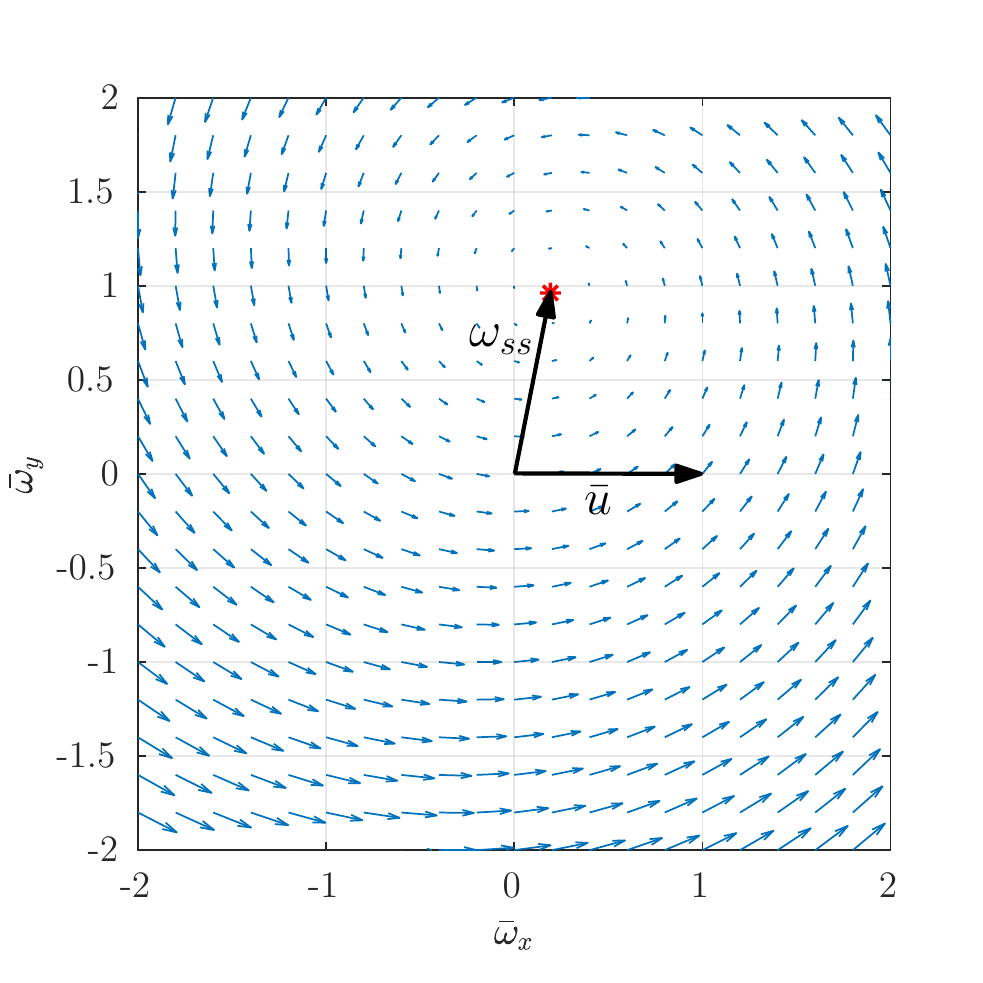}
	\caption{$\bar{u} = [u_0,0]$}
	\label{fig:tor_x}
\end{subfigure} \qquad
\begin{subfigure}[h!]{.3\textwidth}
\centering
	\includegraphics[width=1.1\textwidth]{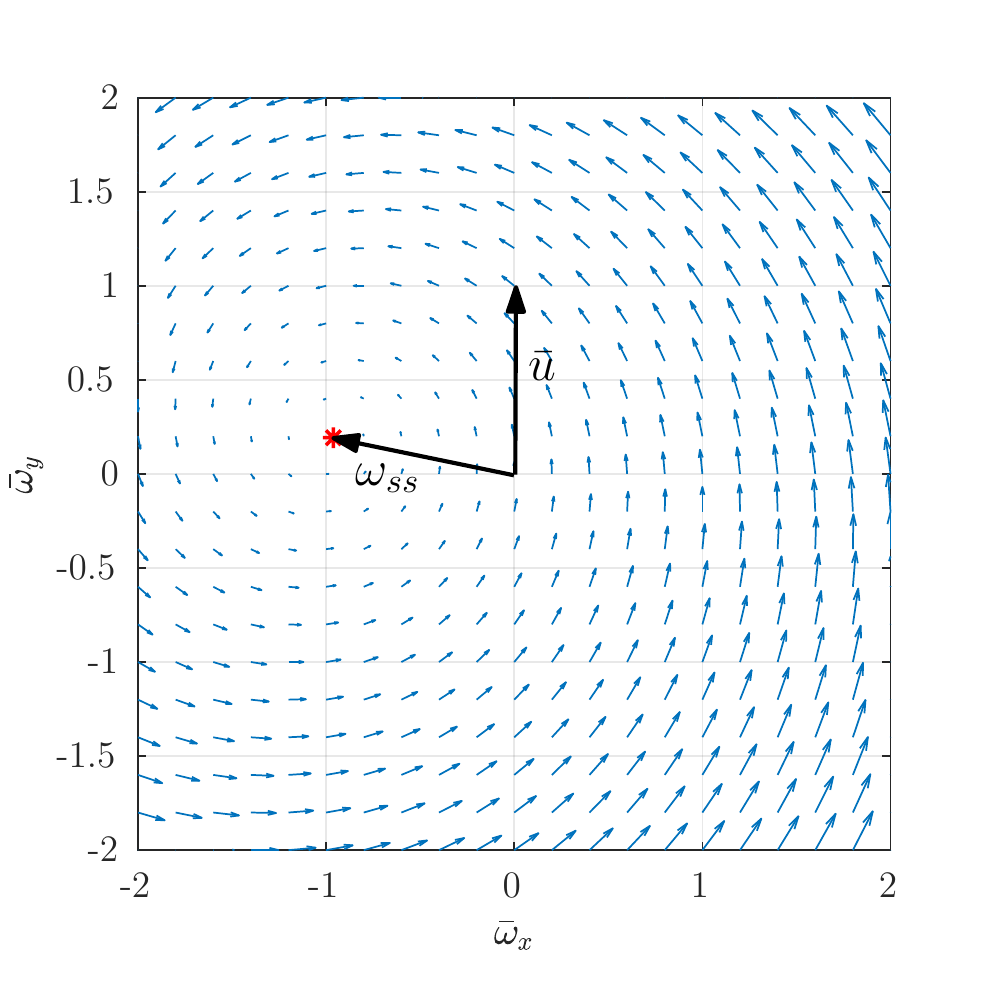}
	\caption{$\bar{u} = [0,u_0]$}
	\label{fig:tor_y}
\end{subfigure}
\caption{Phase portrait illustrating the gyroscopic effect. Here $\omega_{ss}$: steady state angular velocity, $\bar{u}$: input torque.}
\label{fig:gyro_phase}
\end{figure}

From Fig. \ref{fig:gyro_phase}, it is evident that there is almost 90 degree lag in angular velocity response to a constant torque input given in frame $F_1$. This particular aspect would be taken into account while deriving the reduced attitude controller for the trispinner in Section \ref{sec:controler}. Theoretically, it is possible to cancel the nonlinearity arising from the gyroscopic term. However, the gyroscopic term is large in magnitude and due to the delay associated with the sensors and actuators in practical implementation it is not perfectly cancellable. Further, the gyroscopic term adds inherent passive stability to the system. The main idea in this work is to preserve the gyroscopic stability during the design phase so that the closed loop system inherits it.

\section{Reduced Attitude Controller}\label{sec:controler}
In this section we present the structure preserving reduced attitude controller and the conventional reduced attitude controller. Both of them are shown to be almost globally asymptotically stable. 
The governing equations of the reduced attitude in body fixed frame $F_2$ given by \eqref{eq:kin} and \eqref{eq:gyro} are reproduced here for convenience
\begin{equation} \label{eq:dyn_gy}
\begin{gathered}
\dot{\Gamma} = \Omega \times \Gamma  \\
\dot{\omega} = A_{skw} \omega + u.
\end{gathered}
\end{equation}

\subsection{Structure Preserving Controller}
The proposed structure-preserving reduced attitude controller is given by 
\begin{equation} \label{eq:control_gy}
u = -A_{sym} \omega - A \omega_d + \dot{\omega}_d, 
\end{equation}
where $\omega_d \triangleq E_2 R_2^T \Omega_d \in \mathbb{R}^2$, $E_2$ is the $2\times 3$ projection matrix which picks xy components, $\Omega_d \triangleq k_P \Gamma \times \Gamma_d$, $A_{sym} \triangleq diag(k_D, k_D)$, and $A \triangleq A_{skw} - A_{sym}$. Note that when $\Gamma$ and $\Gamma_d$ are not collinear, then $\omega_d$ points in the direction of the geodesic joining the two points on $S^2$.
With the control input \eqref{eq:control_gy} substituted in \eqref{eq:dyn_gy}, the following error dynamics is obtained
\begin{equation} \label{eq:err_dyn}
\begin{gathered}
\dot{\Gamma} = (\Omega_d + \Omega_e) \times \Gamma \\
\dot{\omega}_e = A \omega_e,
\end{gathered}
\end{equation}
where $\Omega_e \triangleq \Omega - \Omega_d$. One may verify from the matrix $A$ that the closed loop error dynamics is that of a damped gyroscope, and therefore it is structure-preserving as well as stable. 

The closed loop system has two equilibria: $\chi \triangleq \{(\Gamma_d,0), (-\Gamma_d,0)\}$. Next we show the stability results of the proposed controller. 

\begin{theorem} \label{thm:gyro}
For $k_P > 0$, $k_D > 1/4$, the control input \eqref{eq:control_gy} renders the desired equilibrium $(\Gamma_d,0)$ of the closed loop system \eqref{eq:err_dyn} almost globally asymptotically stable.
\end{theorem}

\begin{proof}
Consider the following configuration error function on the sphere
\begin{equation} \label{eq:err_fun}
\Psi(\Gamma) = 1 - \Gamma_d^T \Gamma.
\end{equation}
Its derivative along the closed loop vector field \textcolor{blue}{\eqref{eq:err_dyn}} is given by 
\begin{equation}
\dot{\Psi}(\Gamma) = - \Gamma_d^T ((\Omega_d + \Omega_e) \times \Gamma) =  \frac{-1}{k_P}(\norm{\Omega_d}^2 + \Omega_d^T\Omega_e) = \frac{-1}{k_P}(\norm{\omega_d}^2 + \omega_d^T\omega_e),
\end{equation}
where we have used the fact that a rotation matrix preserves the inner product. Next, consider the following candidate Lyapunov function 
\begin{equation}
V(\Gamma, \omega_e) = k_P\Psi + \half \omega_e^T \omega_e.
\end{equation}
The directional derivative of $V$ along the closed loop vector field is given by 
\begin{equation} 
\begin{aligned}
\dot{V} &= -(\norm{\omega_d}^2 + \omega_d^T\omega_e) + \omega_e^T A \omega_e = -(\norm{\omega_d}^2 + \omega_d^T\omega_e) - k_D \norm{\omega_e}^2 \\
\dot{V} & \leq -\norm{\omega_d}^2 + \norm{\omega_d}\norm{\omega_e} - k_D \norm{\omega_e}^2 \\
\dot{V} & \leq  \begin{bmatrix}
\norm{\omega_d} & \norm{\omega_e}
\end{bmatrix}
\underbrace{
\begin{bmatrix}
-1 & \frac{1}{2} \\
\frac{1}{2} & -k_D
\end{bmatrix}}_{Q}
\underbrace{
\begin{bmatrix}
\norm{\omega_d} \\ \norm{\omega_e}
\end{bmatrix}}_{X}.
\end{aligned}
\end{equation}
It could be verified that for a given $k_P > 0$, $Q$ is negative definite if $k_D > 1/4$.
Therefore, 
\begin{equation}
\dot{V} \leq X^T Q X < 0
\end{equation}
for $X \neq 0$. Therefore, all initial conditions in the state space converge to the set characterized by $X \equiv 0$, which is precisely the set \textcolor{blue}{$\chi$}  of two equilibrium points of the closed loop system. It could be shown with linearizion about these points that both the points are hyperbolic and the desired equilibrium point $(\Gamma_d, 0)$ is stable while the undesired one $(-\Gamma_d, 0)$ is unstable (see \ref{appd:lin}). Hence the stable manifold associated with the undesired equilibrium is of measure zero. Therefore, the set of all the initial conditions in the state space that lie outside the stable manifold of the undesired equilibrium converges to the desired equilibrium.
\end{proof}

\subsection{Conventional Controller}

The conventional reduced attitude controller is given by 
\begin{equation} \label{eq:control_conv}
u = -A_{sym} \omega + \omega_d,
\end{equation}
which when substituted in \eqref{eq:dyn_gy} results in the following error dynamics
\begin{equation} \label{eq:err_dyn1}
\begin{gathered}
\dot{\Gamma} = \Omega \times \Gamma \\
\dot{\omega} = A_{skw} \omega - A_{sym}\omega + \omega_d.
\end{gathered}
\end{equation}
The following theorem shows the stability properties of the conventional controller. The set of equilibrium points of the above error dynamics are the same as the previous case, i.e. $\chi$.
\begin{theorem}
For $k_P > 0$, $k_D > 0$, the control input \eqref{eq:control_conv} renders the desired equilibrium $(\Gamma_d,0)$ of the closed loop system \eqref{eq:err_dyn1} almost globally asymptotically stable.
\end{theorem}

\begin{proof}
Consider the same configuration error function on the sphere as in the previous case, $\Psi(\Gamma)=1-\Gamma_d^{T}\Gamma$, whose directional derivative along the closed loop vector field \eqref{eq:err_dyn1} is given by
\begin{equation}
\dot{\Psi}(\Gamma) = - \Gamma_d^T (\Omega \times \Gamma) = - \Omega^T ( \Gamma \times \Gamma_d).
\end{equation}
Now consider the following candidate Lyapunov function 
\begin{equation}
V = k_P \Psi + \half \omega^T \omega
\end{equation}
with its directional derivative along \eqref{eq:err_dyn1} being given by
\begin{equation}
\begin{aligned}
\dot{V} &= k_P \dot{\Psi} + \omega^T \dot{\omega} = -k_P\Omega^T ( \Gamma \times \Gamma_d) + \omega^T (A_{skw}\omega - A_{sym}\omega + \omega_d)\\
&= -k_P\Omega^T ( \Gamma \times \Gamma_d) -k_D \norm{\omega}^2 + \omega^T \omega_d\\
 &= -k_D \norm{\omega}^2 \\
 & \leq 0,
\end{aligned}
\end{equation}
where we have used the fact that $\Omega = R_2 E_2^T \omega$ and $\omega_d = k_P E_2 R^T (\Gamma \times \Gamma_d)$. Since $V$ is bounded from below, all the initial conditions would converge to the set characterized by $\dot{V} \equiv 0$, i.e. $\norm{\omega} \equiv 0$. From \eqref{eq:err_dyn1}, it is evident that the positive limit set is $\chi$.  Linearization reveals the local structure around each equilibrium (see \ref{appd:lin}). As in the previous case, both the equilibria are hyperbolic and the desired equilibrium is stable whereas the undesired one is unstable. Therefore, all the initial conditions in the state space which start outside the stable manifold of the undesired equilibrium converge to the desired equilibrium. 

\end{proof}

In Section \ref{sec:sim}, we will show that the conventional reduced attitude controller, although provably almost globally asymptotically stable, is at the verge of instability in case of spin axis control of a gyroscope. It is also very inefficient in terms of control action when compared to the structure preserving controller.

\subsection{Structure Preserving Controller with Motor Dynamics and Observer}
In a real-world implementation, the required torque by the controller cannot be realized instantly due to actuator dynamics. This is particularly important in case of a fast spinning body as the body frame torque demanded by the controller is cyclic and  a time lag in its application could result in inefficient control action and even instability. In this work, the motor dynamics is modeled as a first order system given by
\begin{equation} \label{eq:motor_dy}
\dot{u} = A_m (u - v),
\end{equation}
where, $A_m = diag(-1/\tau_m,-1/\tau_m)$, $\tau_m$ is the motor time constant, and $v$ is the new control input. Now the system to be controlled is given by \eqref{eq:dyn_gy} and \eqref{eq:motor_dy}. The structure-preserving controller which compensates for the motor dynamics is given by,
\begin{equation} \label{eq:control_gy_mot}
v = u_d - A_m^{-1} \dot{u}_d, 
\end{equation}
where $u_d = -A_{sym} \omega - A \omega_d + \dot{\omega}_d$ is the desired torque given by \eqref{eq:control_gy}. The above control input results in the following error dynamics
\begin{equation} \label{eq:err_dyn_mot}
\begin{gathered}
\dot{\Gamma} = (\Omega_d + \Omega_e) \times \Gamma, \\
\dot{\omega}_e = A \omega_e + u_e, \\
\dot{u}_e = A_m u_e,
\end{gathered}
\end{equation}
where $u_e \triangleq u - u_d$. The set of equilibrium points associated with the above vector field is given by $\chi_m = \{ (\Gamma_d,0,0), (-\Gamma_d,0,0) \}$. The following theorem gives the stability properties of the closed loop system.
\begin{theorem} \label{thm:gyro_mot}
For $k_P>0$, $k_D > (1 + \tau_m)/4$, the control input \eqref{eq:control_gy_mot} renders the desired equilibrium, $(\Gamma_d,0,0)$, of the closed loop system \eqref{eq:err_dyn_mot} almost globally asymptotically stable.
\end{theorem}
\begin{proof}
Consider the following candidate Lyapunov function
\begin{equation}
V = k_P\Psi + \half \omega_e^T \omega + \half u_e^T u_e.
\end{equation}
The derivative of $V$ along the closed loop vector field is given by 
\begin{equation}
\begin{aligned}
\dot{V} &= k_P \dot{\Psi} + \omega_e^T \dot{\omega}_e + u_e^T \dot{u}_e  \\
 &= -\norm{\omega_d}^2 - \omega_d^T\omega_e + \omega_e^T A \omega_e + \omega_e^T u_e + u_e^T A_m u_e \\ 
 &= -\norm{\omega_d}^2 - \omega_d^T\omega_e - k_D \norm{\omega_e}^2 + \omega_e^T u_e - \frac{1}{\tau_m} \norm{u_e}^2, 
 \end{aligned}
\end{equation}
 then, 
 \begin{equation}
\begin{aligned}
\dot{V} & \leq -\norm{\omega_d}^2 + \norm{\omega_d}\norm{\omega_e} - k_D \norm{\omega_e}^2 + \norm{\omega_e}\norm{u_e} - \frac{1}{\tau_m} \norm{u_e}^2 \\
\dot{V} & \leq  \begin{bmatrix}
\norm{\omega_d} & \norm{\omega_e} & \norm{u_e}
\end{bmatrix}
\underbrace{
\begin{bmatrix}
-1 & \half & 0 \\
\half & -k_D & \half \\
0 & \half & -\frac{1}{\tau_m}
\end{bmatrix}}_{Q}
\underbrace{
\begin{bmatrix}
\norm{\omega_d} \\ \norm{\omega_e} \\ \norm{u_e}
\end{bmatrix}}_{X}.
\end{aligned}
\end{equation}
The matrix $Q$ is negative definite if all the upper left diagonal submatrices of $-Q$ have positive determinant, which leads to the condition $k_D > (1 + \tau_m)/4$. Therefore, $\dot{V} \leq X^T Q X < 0$ for $X \neq 0$. Hence, the positive limit set of all initial conditions in the state space $S^2 \times \mathbb{R}^2 \times \mathbb{R}^2$ is characterized by $\dot{V} \equiv 0$, i.e. $X \equiv 0$, which is the set $\chi_m$. Linearization of the error dynamics \eqref{eq:err_dyn_mot} about the two equilibrium points in $\chi_m$ reveal that the desired equilibrium $(\Gamma_d,0,0)$ is stable, while the other one is unstable (see \ref{appd:lin}). Using similar arguments as in the previous case, the stable manifold of $(-\Gamma_d,0,0)$ is of measure zero. Hence all initial conditions that start outside the stable manifold of the undesired equilibrium converge to $(\Gamma_d,0,0)$.

\end{proof}

\begin{remark}
Note that the condition $k_D > (1+\tau_m)/4$ given in previous theorem is a mild one and is close to the one given in Theorem \ref{thm:gyro} as the actuator time constant is  small, $\tau_m << 1$.
\end{remark}

The controller proposed in \eqref{eq:control_gy_mot} requires calculation of $\dot{u}_d$ which is dependent on the angular acceleration $\dot{\omega}$, a quantity not directly measurable using existing set of sensors. To circumvent this issue, we design an observer for $u$ and obtain $\dot{\omega}$ from \eqref{eq:dyn_gy}. Since the actuator dynamics given by \eqref{eq:motor_dy} is already stable and fast, the observer could be simply chosen as 
\begin{equation}
\dot{\hat{u}} = A_m (\hat{u} - v),
\end{equation}
which results in the following observer error dynamics
\begin{equation}
\dot{\tilde{u}} = A_m \tilde{u},
\end{equation}
where $\tilde{u} \triangleq u_d - u$. The modified controller with the observer is given by
\begin{equation} \label{eq:control_gy_mot_obs}
v = u_d - A_m^{-1} \dot{\bar{u}}_d,
\end{equation}
where
\begin{equation}
\begin{gathered}
\dot{\bar{u}}_d = -A_{sym}\dot{\bar{\omega}} - A\dot{\omega}_d + \ddot{\bar{\omega}}_d, \\
\ddot{\bar{\omega}}_d = k_P \hat{e}_3 ( \hat{\omega}^2 R_2^T R_d e_3 - (R_2^T R_d e_3)^\wedge \dot{\bar{\omega}}), \\
\dot{\bar{\omega}} = A_{skw}\omega + \hat{u}.
\end{gathered}
\end{equation}
The above controller results in the following closed loop system
\begin{equation} \label{eq:err_dyn_mot_obs}
\begin{gathered}
\dot{\Gamma} = \Omega \times \Gamma, \\
\dot{\omega}_e = A\omega_e + u_e, \\
\dot{u}_e = A_m u_e + M \tilde{u}, \\
\dot{\tilde{u}} = A_m \tilde{u},
\end{gathered}
\end{equation}
where $M \triangleq A_{sym} + k_P \hat{e}_3 (R_2^T R_d e_3)^\wedge$. The set of equilibrium is denoted $\chi_o \triangleq \{ (\Gamma_d,0,0,0), (-\Gamma_d,0,0,0) \}$.
The following theorem gives the stability properties of $\chi_o$ associated with the closed loop system.
\begin{theorem}
For $k_P > 0$, $k_D > (1+\tau_m)/4$, the controller \eqref{eq:control_gy_mot_obs} renders the desired equilibrium $(\Gamma_d,0,0,0)$ of the closed loop system \eqref{eq:err_dyn_mot_obs} almost globally asymptotically stable.
\end{theorem}
\begin{proof}
Since the observer error dynamics is independent of the controller, linear and stable, $\lim_{t \to \infty} \tilde{u} \to 0$ exponentially. Further, $M$ is bounded for all time as $R_2^T R_d \in SO(3)$ and $SO(3)$ is compact. As a result $\omega_e$ and $u_e$ are bounded as $A$ and $A_m$ are Hurwitz. Since $\Gamma$ evolves on the compact manifold $S^2$, it also remains bounded. In Theorem \ref{thm:gyro_mot}, the desired equilibrium $(\Gamma_d,0,0,0)$ of the error dynamics \eqref{eq:err_dyn_mot_obs} has been shown to be almost globally asymptotically stable with zero observer error, i.e. $\tilde{u} \equiv 0$.
\end{proof}
A proof for showing almost global stability of a controller-observer system evolving on a Lie group is given in \cite{maithripala2006almost} (Lemma 2). 

\section{Simulation Results} \label{sec:sim}

In this section we give a comparison of the various reduced attitude controllers presented in this work with numerical simulations. The simulations are carried out for spin axis stabilization of a gyroscope whose  parameters are provided in Table \ref{tab:spinner_params}. 

\begin{table} [h!]
\renewcommand{\arraystretch}{1.3}
\caption{Trispinner parameters}
\label{tab:spinner_params}
\centering
\begin{tabular}{l l l} 
 \hline
 Parameter & Description & Values  \\
 \hline
 $[J_{xx}, J_{yy}, J_{zz}]$ & Moment of inertia & [1.55, 1.55, 2.76]$\times10^{-2}$ $kg$-$m^2$  \\ 
 
 $m$ & Vehicle mass & 490 $g$\\ 
 
 $\tau_m$  & Motor time constant & 0.0846 $s/rad$  \\  

 $L$  & Motor arm length & 13.4 $cm$ \\
 
 $F_{max}$ & Maximum motor force & 6.74 $N$ \\ [1ex]
 \hline
\end{tabular}
\end{table}

The torque actuators of the gyroscope is approximated to be first order and are rigidly fixed to the body frame and hence spin with the gyroscope. A Lie group variational integrator \cite{lieintegrator2006} is used to simulate the rigid body dynamics of the gyroscope. In all the simulations the reduced attitude $\Gamma$ of the spin axis is plotted on a unit sphere in blue to show the path it takes to reach the desired orientation $\Gamma_d$. This visualization helps in comparing different controllers in terms of stability and the efficiency in terms of geodesic distance, i.e. closeness of the path traced out by the spin axis on the sphere to the geodesic connecting the initial and desired orientation. For all the simulations, the initial spin axis orientation is horizontal which is denoted with blue unit vector on the sphere plot, the desired spin axis orientation is vertical and is denoted with an orange unit vector and the final orientation of the spin axis at the end of the simulation is denoted with a red unit vector.

First, a comparison is made between the conventional reduced attitude controller given by \eqref{eq:control_conv} with the structure preserving controller given by \eqref{eq:control_gy} without the motor dynamics and is shown in Fig. \ref{fig:sim_no_motor}. Both the controllers were tuned with the same set of gains.

\begin{figure}[h!]
\centering
\begin{subfigure}{.4\textwidth}
\centering
	\includegraphics[width=1\textwidth]{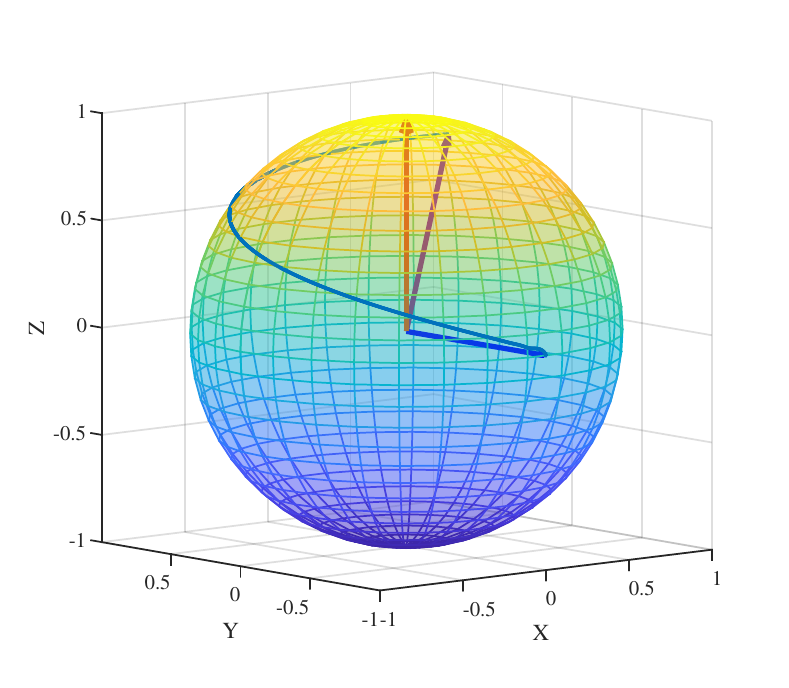}
\end{subfigure} \qquad
\begin{subfigure}{.4\textwidth}
\centering
	\includegraphics[width=1\textwidth]{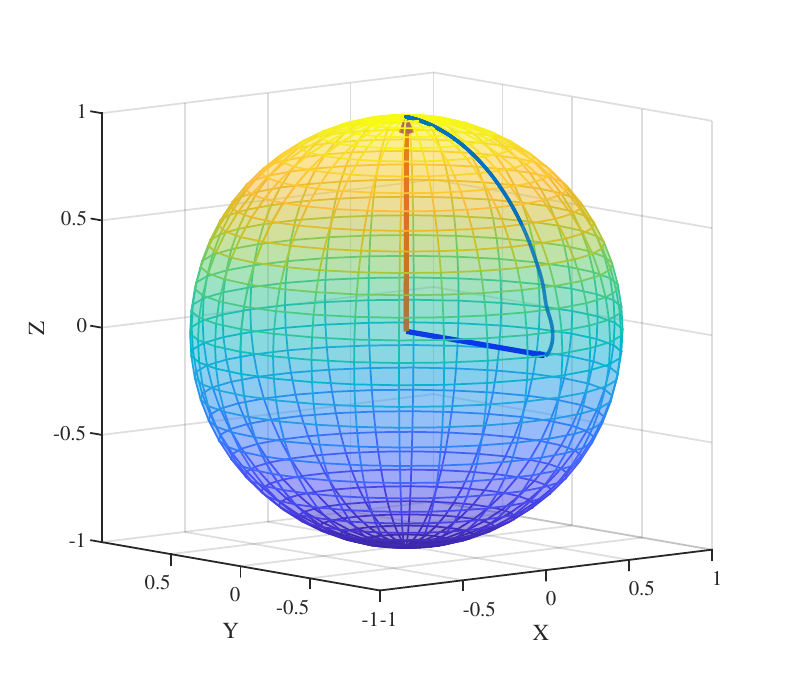}
\end{subfigure}
\begin{subfigure}{.5\textwidth}
\centering
	\includegraphics[width=1.05\textwidth]{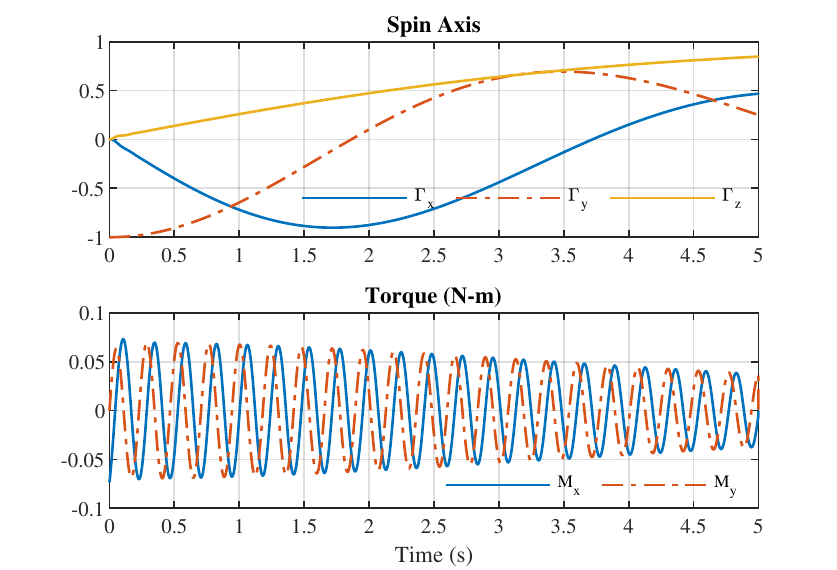}
	\caption{Conventional Controller}
	\label{fig:sim_nomot_conv}
\end{subfigure}%
\begin{subfigure}{.5\textwidth}
\centering
	\includegraphics[width=1.05\textwidth]{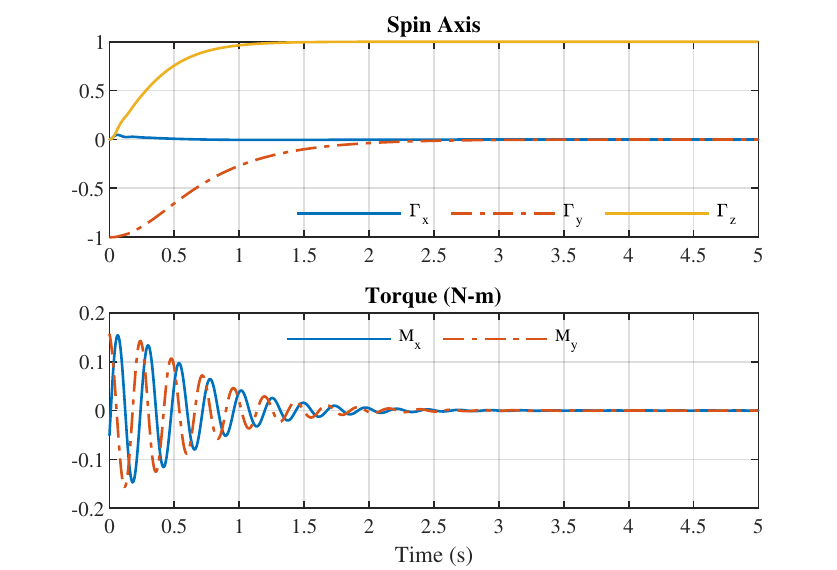}
	\caption{Structure Preserving Controller}
	\label{fig:sim_nomot_str}
\end{subfigure}
\caption{Comparison of the conventional and the structure preserving controller in simulation without the motor dynamics.}
\label{fig:sim_no_motor}
\end{figure}

Since the structure preserving controller explicitly takes into account the gyroscopic torque, the path traced out by the spin axis on the sphere is close to the geodesic. Note that the controller does not cancel the gyroscopic torque, instead it is retained in the closed loop such that the entire system inherits the stability associated with the term. However, the conventional reduced attitude controller does not acknowledge the large gyroscopic torque, and applies the proportional torque along the geodesic direction. This results in the spin axis tracing out a curved path on the sphere far away from the geodesic. For the given simulation time of 5 seconds, the structure preserving controller converges in about 2 seconds whereas the conventional controller has an error of 32 deg. between the final and the desired spin axis orientation.

In Fig. \ref{fig:sim_with_motor}, we compare the performance of the conventional and the structure preserving controllers in the presence of first order motor dynamics. The former is unstable and diverges from the desired spin axis $\Gamma_d$ to it's antipodal point $-\Gamma_d$. Whereas the structure preserving controller, although asymptotically Lyapunov stable, has a very slow rate of convergence for the same set of gains as used in the previous case. The final spin axis orientation is off by 12.5 deg. from the desired vertical direction after 15 sec. while tracing out a spiral on the sphere. 

Finally, in Fig. \ref{fig:sim_with_motor_obs}, we show the performance of the structure preserving controller with the observer  which explicitly takes into account the motor dynamics as given by \eqref{eq:control_gy_mot_obs}. The performance is similar to that of the controller given by \eqref{eq:control_gy} in the case without motor dynamics as shown in Fig. \ref{fig:sim_nomot_str}, except for the initial transient due to the observer dynamics. The spin axis \textcolor{blue}{achieves} the desired vertical direction in 2 seconds while being close to the geodesic connecting the initial and the desired spin axis direction.

Note that the simulations study conducted in this section shows that the conventional controller, although proven to be almost globally asymptotically stable, is prone to be unstable in the presence of motor dynamics. The structure preserving controller which does not take into account the motor dynamics is stable but takes a long time to converge to the equilibrium.

\begin{figure}[h!]
\centering
\begin{subfigure}{.4\textwidth}
\centering
	\includegraphics[width=1\textwidth]{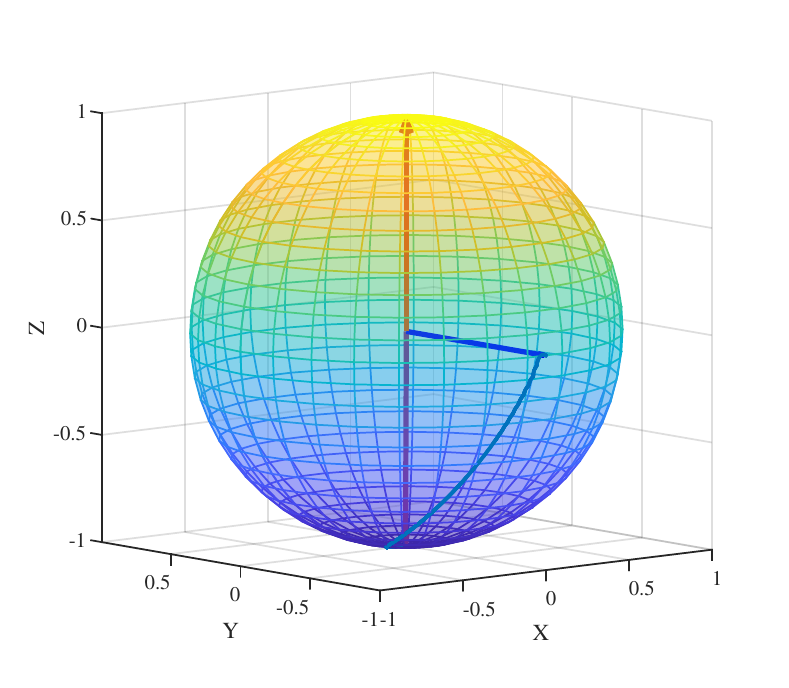}
\end{subfigure} \qquad
\begin{subfigure}{.4\textwidth}
\centering
	\includegraphics[width=1\textwidth]{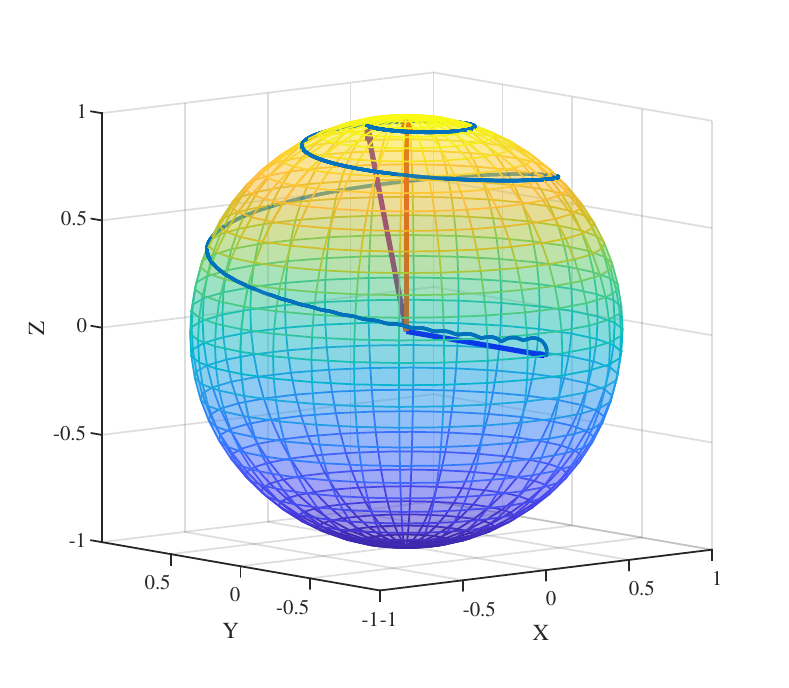}
\end{subfigure}
\begin{subfigure}{.5\textwidth}
\centering
	\includegraphics[width=1.05\textwidth]{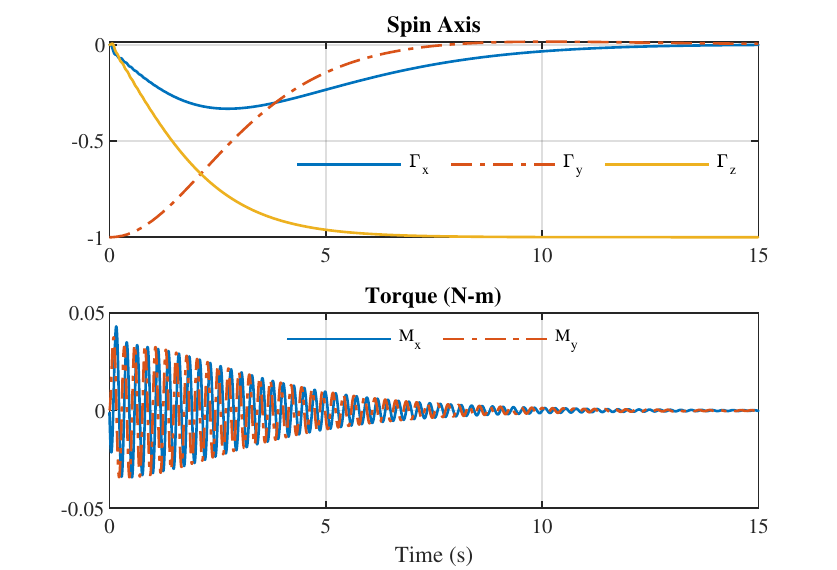}
	\caption{Conventional Controller}
	\label{fig:sim_mot_conv}
\end{subfigure}%
\begin{subfigure}{.5\textwidth}
\centering
	\includegraphics[width=1.05\textwidth]{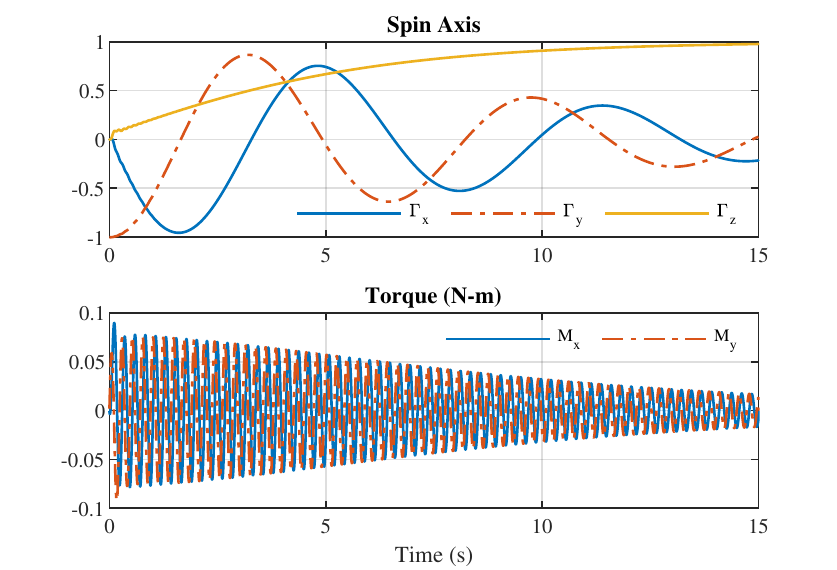}
	\caption{Structure Preserving Controller}
	\label{fig:sim_mot_str}
\end{subfigure}
\caption{Comparison of the conventional and the structure preserving controller in simulation with the motor dynamics.}
\label{fig:sim_with_motor}
\end{figure} \begin{figure}[h!]
\centering
\begin{subfigure}{.4\textwidth}
\centering
	\includegraphics[width=1\textwidth]{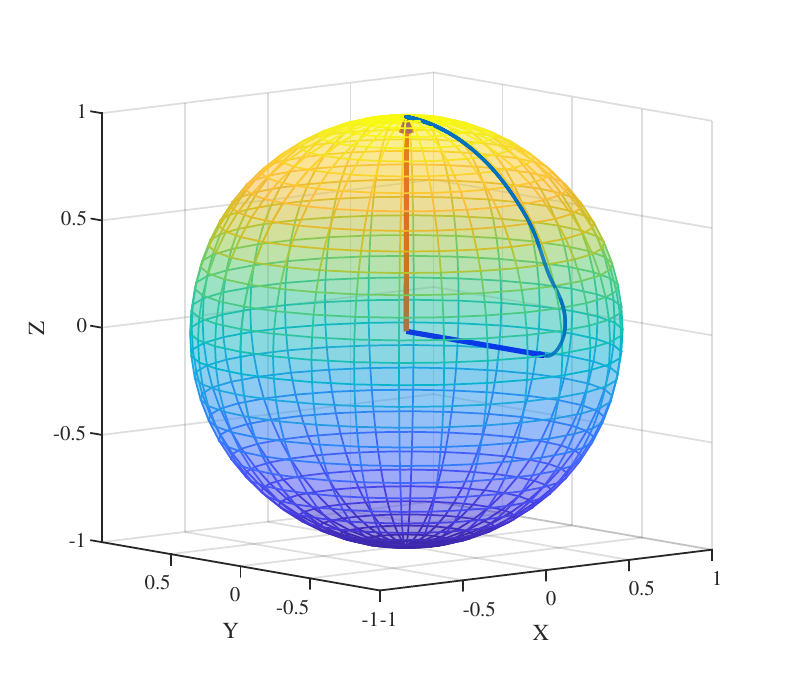}
\end{subfigure}
\begin{subfigure}{.5\textwidth}
\centering
	\includegraphics[width=1.05\textwidth]{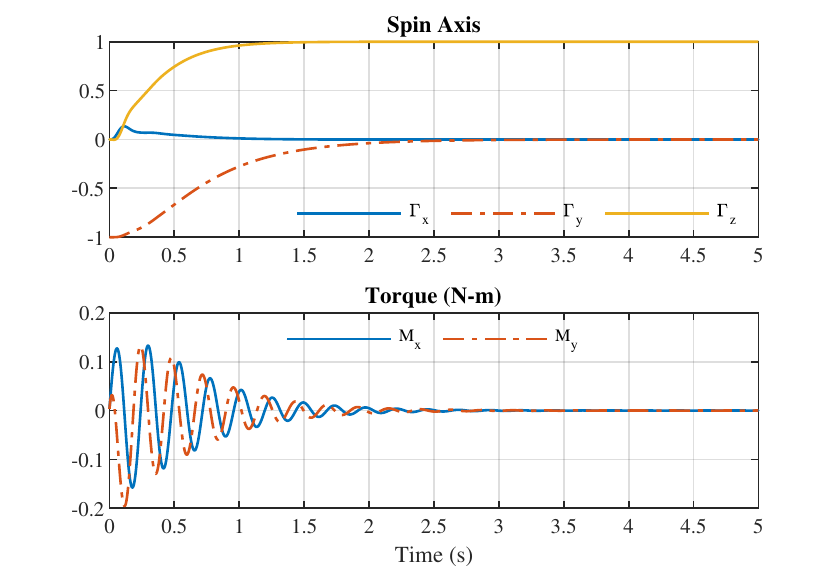}
	\label{fig:sim_motobs_str}
\end{subfigure}
\caption{Performance of the structure preserving controller with observer in simulation with the motor dynamics.}
\label{fig:sim_with_motor_obs}
\end{figure}

\section{Experimental Results}
The experimental validation of the proposed controller is performed on a spinning axisymmetric tricopter, henceforth called trispinner (see Fig. \ref{fig:trispinner}).
The trispinner has three motors with propellers, all spinning in the same direction, as actuators. They are capable of generating net thrust along the spin axis and body fixed torques $M_x$ and $M_y$. Each spinning propeller experiences aerodynamic drag torque about its axis of rotation. The net reaction torque of all the three spinning propellers makes the entire body of the trispinner spin and attain a steady state spin wherein the drag torque experienced by the body balances that of the propeller. 

The vehicle is equipped with Pixhawk autopilot hardware running the PX4 flight stack. The autopilot consists of a triaxial gyro, a triaxial accelerometer and a triaxial magnetometer together constituting the attitude heading reference system (AHRS). The stock quaternion based complimentary filter on PX4 estimates the attitude of the vehicle \cite{mahony2008nonlinear}. There are two important aspects pertaining to the sensor suit, specifically accelerometer and gyro, which are taken into consideration for the experimental implementation. Since the accelerometer is not located exactly on the spin axis, it measures the centripetal acceleration. The role of accelerometer in the attitude estimator is to correct for the gyro drift \cite{accle2010est}. The centripetal acceleration is corrected for in the attitude estimator. Second, since the onboard gyro saturates at 2000 deg/s, the motors are slightly tilted about the motor arm axis to reduce the steady state spin rate to around 1500 deg/s. The proposed structure preserving attitude controller is implemented as a separate module in the autopilot and the control loop runs at 250 Hz. 

The controller is validated by performing manual flight wherein the desired spin axis direction is specified by the pilot using a joystick. The performance of the proposed structure preserving controller with the observer given by \eqref{eq:control_gy_mot_obs} is shown in Fig. \ref{fig:exp_plot} and the experimental flight video is available in the link given in Fig. \ref{fig:trispinner}. 

\begin{figure}[h!]
\centering
\includegraphics[width=0.5\linewidth]{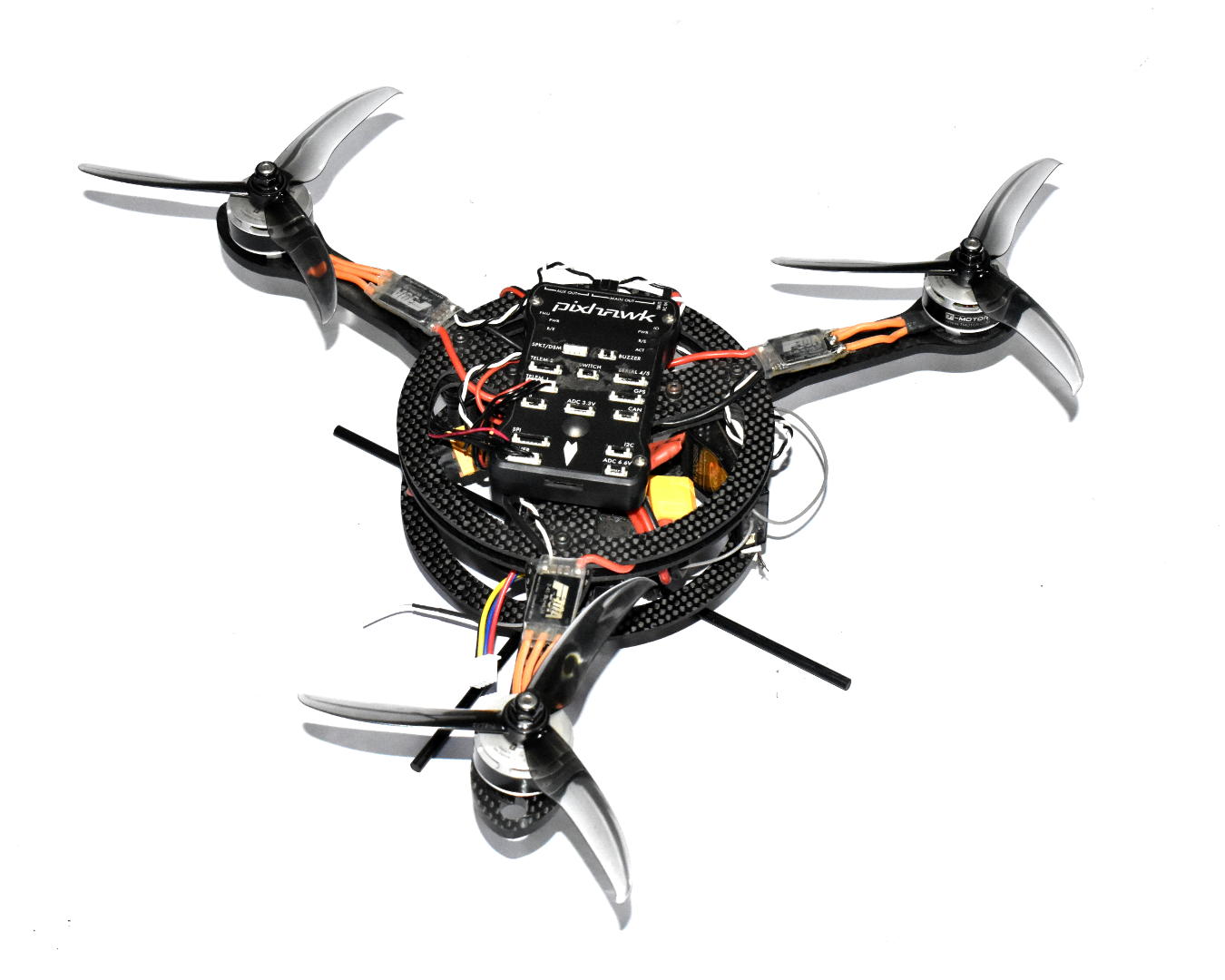}
\caption{Trispinner used for experimental validation of structure preserving controller.\\ Flight video: \url{https://youtu.be/YNm6_6VP0AM} }
\label{fig:trispinner}
\end{figure}

The controller follows the commanded reduced attitude command given by the pilot with reasonable accuracy as is evident from the plots of Fig. \ref{fig:exp_plot}. Note that the controller is designed for stabilization and not tracking, as a result perfect tracking is not expected. Further there is cross axis response, i.e. when $\Gamma_{xd}$ is commanded there is $\Gamma_y$ gets excited and vice-versa. This could be eliminated by designing a tracking controller which applies appropriate feedforward command. The reduced attitude plots $\Gamma$ in Fig. \ref{fig:exp_plot} has small oscillation with the spinning frequency, which could be attributed to slight misalignment of the autopilot board Z-axis with respect to the spin axis. The torque requirement for performing the maneuvers are sufficiently smooth for implementation purpose as is evident from Fig. \ref{fig:exp_plot}.

\begin{figure}[h!]
\centering
\includegraphics[width=0.7\linewidth]{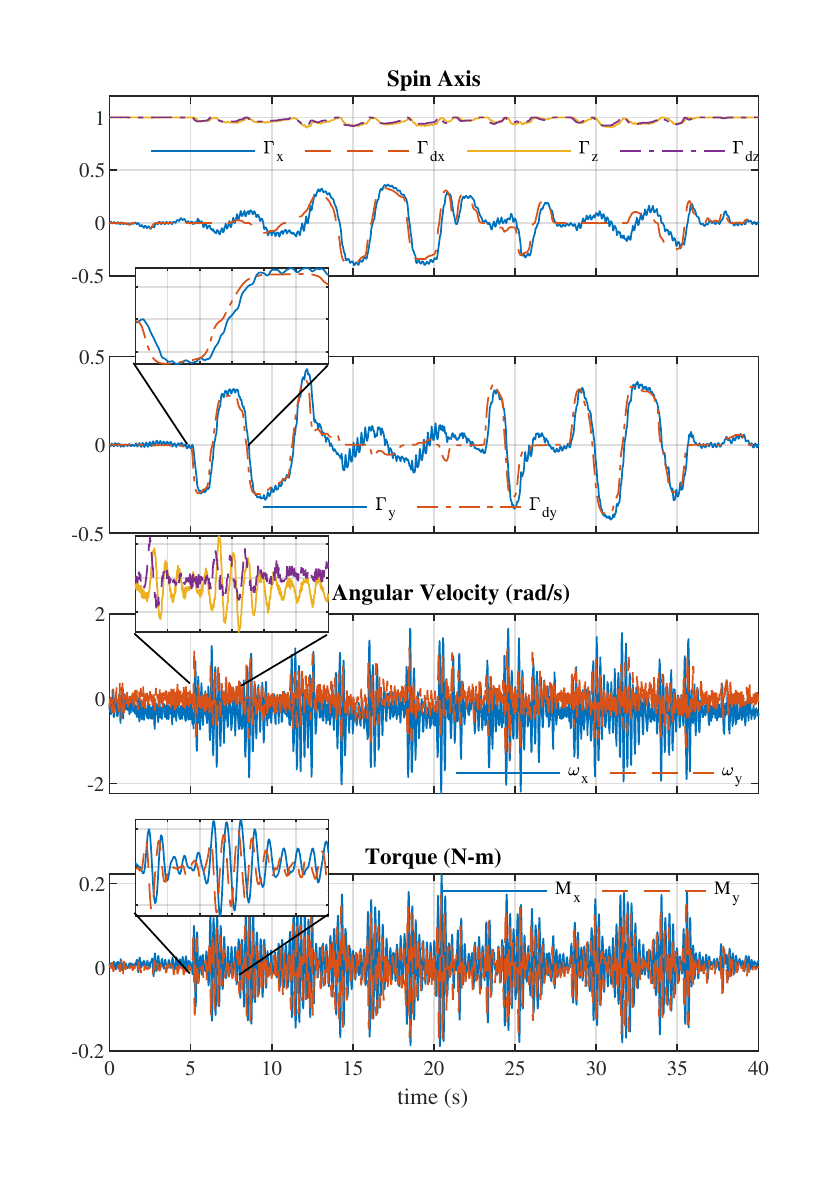}
\caption{Experimental validation of the structure preserving controller with observer}
\label{fig:exp_plot}
\end{figure}

\section{Conclusions}

In this work we have addressed the problem of reorienting the spin axis of a gyroscope in a geometric control framework. The structure preserving reduced attitude controller for a spinning gyroscope proposed in this work has been validated with experiments and simulations in the presence of disturbance. The conventional reduce attitude controller, although shown to be almost-globally asymptotically stable, is close to instability as shown in the simulations study in Sec. \ref{sec:sim}. This is due to the fact that the controller design procedure primarily focused on the kinematics on the underlying configuration space and disregarded the dominant gyroscopic effect. The proposed structure preserving controller explicitly accounts for and augments the inherent passive gyroscopic stability associated with a spinning gyroscope. This is accomplished by retaining the gyroscopic term and imparting the gyroscopic structure to the closed loop dynamics. The procedure takes care of the non-linearities of the configuration manifold, the dominant gyroscopic torque, and actuator dynamics. The explicit accountability of actuator delay is particularly important for the case wherein the actuator is rigidly attached to the body and spins with it. The controller is found to be efficient in control input as well as in terms of the path taken by the spin axis on the sphere. Since the controller is intrinsic, it is globally defined and has the best convergence properties. To the best of the authors' knowledge, no control law has been developed which preserves the gyroscopic stability of spinning axis-symmetric rigid body.

\section{Acknowledgement}
\noindent
The authors would like to gratefully acknowledge Prof. Abhishek for providing the experimental facility and equipments for performing the in-flight validation of the controller in Helicopter and VTOL Lab, IIT Kanpur. The trispinner vehicle has been built using resources obtained through the Fund for Improvement of S\&T infrastructure in universities \& higher educational institutions (FIST) grant from the Department of Science and Technology, India. L. Colombo was partially supported by I-Link Project (Ref: linkA20079) from CSIC, by Ministerio de Econom\'ia, Industria y Competitividad (MINEICO, Spain) under grant MTM2016-76702-P and by `Severo Ochoa Programme for Centres of Excellence'' in R\&D (SEV-2015-0554). The project that gave rise to these results received the support of a fellowship from ``la Caixa' Foundation'' (ID 100010434). The fellowship code is
LCF/BQ/PI19/11690016. 

\clearpage
\appendix
\section{Linearization of Error Dynamics} \label{appd:lin}
In order to linearize the error dynamics \eqref{eq:err_dyn} and \eqref{eq:err_dyn1}, we first obtain parametrization of points in the neighborhood the equilibrium points. The perturbation of equilibrium points are given by $(\Gamma_\epsilon, \omega_\epsilon) \triangleq (R_\epsilon e_3, \omega_\epsilon) \triangleq (R_{eq} e^{\epsilon \hat{\bar{\eta}}} e_3, \epsilon \bar{\zeta})$. Here, $(\bar{\eta}, \bar{\zeta}) \in \mathbb{R}^3 \times \mathbb{R}^3 $ are the linearized variables such that $\bar{\eta}_z = \bar{\zeta}_z = 0$ and its $x-y$ components are given by $(\eta, \zeta) \in \mathbb{R}^2 \times \mathbb{R}^2$ such that $(\eta, \zeta) \triangleq (E_2 \bar{\eta}, E_2 \bar{\zeta})$, with $E_2$ being the $2\times 3$ projection matrix. The linearized equations are obtained by substituting them in the error dynamics and differentiating both sides of the equation with respect to $\epsilon$, evaluated at $\epsilon = 0$. Since $\dot{\Gamma}_\epsilon = R_\epsilon \hat{\omega}_\epsilon e_3 = R_{eq}e^{\epsilon \hat{\bar{\eta}}} \epsilon \hat{\bar{\zeta}} e_3$
\begin{equation}
\begin{aligned}
\frac{d}{dt} \frac{d}{d\epsilon} \biggr\rvert_{\epsilon = 0} \Gamma_\epsilon &=
 \frac{d}{d\epsilon} \biggr\rvert_{\epsilon = 0} R_{eq}e^{\epsilon \hat{\bar{\eta}}} \epsilon \hat{\bar{\zeta}} e_3
\\
 \frac{d}{dt} R_{eq} \hat{\bar{\eta}} e_3 &=  R_{eq} \hat{\bar{\zeta}} e_3.
\end{aligned}
\end{equation}
Since the last component of both $\bar{\eta}$ and $\bar{\zeta}$ are zero, from the above we have 
\begin{equation}
\dot{\eta} = \zeta.
\end{equation}
For linearizing the second equation of \eqref{eq:err_dyn} and \eqref{eq:err_dyn1} we calculate the following
\begin{equation}
\begin{aligned}
\frac{d}{d\epsilon} \biggr\rvert_{\epsilon = 0} \omega_d &= \frac{d}{d\epsilon} \biggr\rvert_{\epsilon = 0} k_P E_2 (e_3 \times R_\epsilon^T R_d e_3) = k_P E_2 \hat{e}_3 (R_{eq}^T R_d e_3)^{\wedge} \bar{\eta} \\ 
&= \underbrace{k_P E_2 \hat{e}_3 (R_{eq}^T R_d e_3)^{\wedge} E_2^T}_{P(\Gamma_{eq})} \eta = P(\Gamma_{eq}) \eta,
\end{aligned}
\end{equation}
similarly
\begin{equation}
\frac{d}{d\epsilon} \biggr\rvert_{\epsilon = 0} \dot{\omega}_d = k_P E_2 \hat{e}_3 (R_{eq}^T R_d e_3)^{\wedge} E_2^T \zeta = P(\Gamma_{eq}) \zeta.
\end{equation}
The linearization of second equation of \eqref{eq:dyn_gy} after substituting the control input \eqref{eq:control_gy} would result in 
\begin{equation}
\dot{\zeta} = A\zeta -A P(\Gamma_{eq}) \eta + P(\Gamma_{eq})\zeta.
\end{equation}
The linearization of \eqref{eq:err_dyn} could be compactly expressed as 
\begin{equation}
\frac{d}{dt}\begin{bmatrix}
\eta \\ \zeta
\end{bmatrix} =
\underbrace{
\begin{bmatrix}
\mathbf{0_{2\times 2}} & \mathbf{I_{2 \times 2}} \\
 -A P(\Gamma_{eq}) & A + P(\Gamma_{eq}) 
\end{bmatrix}}_{S_1(\Gamma_{eq})}
\begin{bmatrix}
\eta \\ \zeta
\end{bmatrix}.
\end{equation}
Similarly, the linearization of \eqref{eq:err_dyn1} and \eqref{eq:err_dyn_mot} results in
\begin{equation}
\frac{d}{dt}\begin{bmatrix}
\eta \\ \zeta
\end{bmatrix} =
\underbrace{
\begin{bmatrix}
\mathbf{0_{2\times 2}} & \mathbf{I_{2 \times 2}} \\
 P(\Gamma_{eq})  & A  
\end{bmatrix}}_{S_2(\Gamma_{eq})}
\begin{bmatrix}
\eta \\ \zeta
\end{bmatrix}, \quad \text{and}
\end{equation}
\begin{equation}
\frac{d}{dt}\begin{bmatrix}
\eta \\ \zeta_e  \\ w_e
\end{bmatrix} =
\underbrace{
\begin{bmatrix}
 P(\Gamma_{eq}) & \mathbf{I_{2 \times 2}} & \mathbf{0_{2\times 2}} \\
 \mathbf{0_{2\times 2}}  & A  & \mathbf{I_{2 \times 2}} \\
 \mathbf{0_{2\times 2}} & \mathbf{0_{2\times 2}} & A_m
\end{bmatrix}}_{S_3(\Gamma_{eq})}
\begin{bmatrix}
\eta \\ \zeta_e \\ w_e
\end{bmatrix}
\end{equation}
respectively, where $(\zeta_e, w_e)$ is the linearization of $(\omega_e, u_e)$.

All three matrices, $S_1(\Gamma_{eq})$, $S_2(\Gamma_{eq})$, and $S_3(\Gamma_{eq})$ when evaluated at the desired equilibrium point $\Gamma_d$ are Hurwitz, where as they are unstable when evaluated at $-\Gamma_d$.

\bibliography{ref.bib}

\clearpage 

\end{document}